\documentclass[twocolumn,journal]{IEEEtran}

\ifCLASSINFOpdf
\else
   \usepackage[dvips]{graphicx}
\fi
\usepackage{url}

\usepackage{graphicx}
\usepackage{amsmath, amsfonts, amssymb}
\usepackage[colorlinks=true, allcolors=blue]{hyperref}

\usepackage[ruled, lined, linesnumbered, longend]{algorithm2e}

\usepackage{color}
\usepackage{xcolor}

\usepackage{amsthm}
\usepackage{amssymb}

\newtheorem{assumption}{Assumption}
\newtheorem{lemma}{Lemma}
\newtheorem{theorem}{Theorem}

\theoremstyle{remark}
\newtheorem*{remark}{Remark}

\definecolor{webgreen}{rgb}{0,.5,0} 
\definecolor{lightgreen}{rgb}{0.6,.8,0.6} 
\definecolor{lightgray}{rgb}{0.78,0.78,0.78}
\definecolor{red}{rgb}{1,0,0} 

\SetKwInOut{KwIn}{Input}
\SetKwInOut{KwOut}{Output}

\SetCommentSty{mycommfont}

\usepackage{ulem}

\usepackage{soul}

\renewcommand{\emph}{\textit}

\usepackage{multirow}
\usepackage{array}
\usepackage{booktabs}
\usepackage{xfrac}
\usepackage{color}

\usepackage{comment}
\usepackage{hyperref}

\usepackage[compress]{cite}

\usepackage{graphicx}
\usepackage{tikz}
\usetikzlibrary{shapes.misc}
\tikzset{cross/.style={cross out, draw, 
         minimum size=2*(#1-\pgflinewidth), 
         inner sep=0pt, outer sep=0pt}}
 \usetikzlibrary{calc}         
\usepackage{pgfplots}
\usetikzlibrary{3d}

\definecolor{lightgreen}{rgb}{.9,1,.9}
\definecolor{red}{rgb}{1,0,0}
\definecolor{blue}{rgb}{0,0,1}

\def\lim{\mathop{\mathsf{lim}}} 

\def\log{\mathrm{log}}

\newcommand{\by}{\mathbf{y}}

\newcommand{\bx}{\mathbf{x}}

\newcommand{\bz}{\mathbf{z}}

\newcommand{\bu}{\mathbf{u}}

\begin{document}

\title{Jensen-Bregman Split Gibbs Sampling}

\author{Nicolas Dobigeon,~\IEEEmembership{Senior Member,~IEEE}, Mame Diarra Fall and Eric Barat

\thanks{Part of this work was supported by the AI Interdisciplinary Institute ANITI funded by the France 2030 program (ANR-23-IACL-0002) and the BACKUP project (ANR-23-CE40-0018-01).}
\thanks{N. Dobigeon is with Univ. de Toulouse, IRIT/INP-ENSEEIHT, 31071 Toulouse, France (email:  nicolas.dobigeon@enseeiht.fr).}
\thanks{M. D. Fall is with LITIS, Université Rouen Normandie, Rouen, France, France (email: diarra.fall@univ-rouen.fr).}
\thanks{E. Barat is with Université Paris-Saclay, CEA, List, Palaiseau, France 
  (email: eric.barat@cea.fr.}
}

\markboth{Dobigeon \MakeLowercase{\textit{et al.}}: Jensen-Bregman Split Gibbs Sampling}
{Dobigeon \MakeLowercase{\textit{et al.}}: Jensen-Bregman Split Gibbs Sampling}
\maketitle

\begin{abstract}
Split Gibbs sampling (SGS) is a versatile tool to sample from complex posterior distributions, such as those arising in imaging inverse problems. However, existing SGS instances are limited to quadratic variable couplings, restricting their applicability to Gaussian likelihoods. This letter genearalizes SGS to non-Euclidean geometry by introducing a coupling kernel defined through a Jensen-Bregman divergence. This framework not only enables posterior sampling in these settings but also unifies several \emph{ad hoc} sampling schemes from the literature under a principled theoretical foundation.
\end{abstract}


\IEEEpeerreviewmaketitle

\section{Introduction}\label{sec:introduction}
\IEEEPARstart{S}{ampling} from complex probability distributions $\pi(\bx)$ with composite potentials remains a central challenge in modern statistical inference and machine learning. To handle high-dimensional, non-conjugate, or implicitly defined models, split Gibbs sampling (SGS)  has emerged as a powerful Markov chain Monte Carlo (MCMC) algorithm able that exploits the structure of composite potentials for efficient exploration of the parameter space \cite{Vono2019split,vono2022efficient}. Inspired by splitting-based optimization techniques -- such as the Alternating Direction Method of Multipliers (ADMM) and half-quadratic splitting (HQS) \cite{neal2011distributed,rendell2020global} -- SGS decomposes the sampling task into simpler, tractable subproblems. This decomposition not only facilitates distributed inference, akin to consensus ADMM \cite{rendell2020global}, but also unlocks new avenues for scalable Bayesian inference. In this context, SGS enables practitioners to sample from posterior distributions of the form
\begin{equation}\label{eq:posterior}
    \pi(\bx) \propto \exp\left( - f(\bx;\by) - g(\bx)\right), \ \bx\in\mathbb{R}^m
\end{equation}
where $f(\cdot;\by)$ stands for a data-fidelity term associated with observations $\by \in \mathbb{R}^n$ and $g(\cdot)$ encodes a regularization potential. By leveraging conditional independence with a specifically introduced auxiliary $\bz\in\mathbb{R}^m$, the canonical quadratic splitting-based SGS (Q-SGS) draws samples $\left\{\bx^{(t)},\bz^{(t)}\right\}_{t\geq 1}$ following the iterative scheme
\begin{subequations}\label{eq:SGS_canonical}
\begin{align}
    \bx &\sim \pi_{\rho}(\bx|\bz) \propto \exp\left( - f(\bx;\by) - \tfrac{1}{2\rho}\left\|\bx - \bz\right\|_2^2\right) \label{eq:likelihood_step}\\
    \bz &\sim \pi_{\rho}(\bz|\bx) \propto \exp\left( - g(\bz) - \tfrac{1}{2\rho}\left\|\bx - \bz\right\|_2^2\right) \label{eq:prior_step}.    
\end{align}
\end{subequations}
This scheme strikingly decouples the two potentials $f(\cdot;\by)$ and $g(\cdot)$ into  two distinct conditional distributions handled separately, commonly referred to as the \emph{likelihood step} and the \emph{prior step}, respectively. Beyond its standard implementation, Q-SGS admits insightful interpretations and natural extensions that further enhance its versatility \cite{pereyra2023split}. For instance, the integration of Langevin or proximal steps within the Gibbs scheme allows for gradient-based exploration of the parameter space \cite{plassier2021dg}, e.g, once the prior distribution admits a computable score function \cite{Faye_2025sub} or a proximable potential \cite{vono2018sparse}. Another archetypal example of this versatility lies in its ability to sample from posterior distributions with implicitly defined priors, paving the way for plug-and-play (PnP) algorithms. Interpreting the prior step \eqref{eq:prior_step} as a stochastic denoising task, PnP-SGS has been firstly proposed in \cite{Coeurdoux2024pnp} and subsequently revisited in \cite{bouman2023generative,wu2024principled,Sun2024provable} as a stochastic counterpart of PnP-HQS embedding various generative diffusion-based models as implicit regularizations. These developments position SGS as a flexible and scalable framework for contemporary Bayesian inference, with applications ranging from high-dimensional regression to imaging inverse problems \cite{Faye2024}. 

All aforementioned instances of SGS  have exclusively adopted the scheme \eqref{eq:SGS_canonical} which explicitly relies on a quadratic coupling between the variables $\bx$ and $\bz$. This choice is particularly well suited for Gaussian likelihood functions, where the quadratic coupling potential is conjugate, yielding a straightforward likelihood step \cite{vono2022high}. However, the underlying Euclidean geometry may fail to capture intrinsic structures arising in problems characterized by non-Gaussian likelihoods. Natural extensions of Q-SGS would be to consider coupling defined through Bregman divergences, which generalize the notion of distance to accommodate a broader class of geometries. Even if seemingly appealing, this naive extension has major limitations which precludes its practical value. First, the generally asymmetric nature of the divergence questions the roles of the two variables $\bx$ and $\bz$. Second, even when these roles are chosen  to preserve  conjugacy and facilitate the likelihood step, the resulting prior step may involve an overall potential with poor regularity properties. Such difficulties are highlighted in \cite{faye2025bregman} when solving a Bayesian Poisson inverse problem. To finally arrive at a sampling algorithm with amenable steps, subtle adaptations were made to the standard SGS scheme. This is very unfortunate because this single example demonstrates that extending SGS to the non-Euclidean case cannot be done systematically. It rather requires some expertise in Bayesian modeling that was not essential in the Gaussian case. 

This paper proposes to restore both the versatility and the systematic nature of SGS implementation in the non-Euclidean geometry. To this end, we first recall that SGS anchors within the frame of asymptotically exact data augmentation (AXDA) \cite{vono2020asymptotically}. Building on this unifying perspective, we depart from Q-SGS \eqref{eq:SGS_canonical} by introducing a non-quadratic splitting strategy. To go beyond the crude choice that would adopt a splitting defined by non-symmetric Bregman divergences, we rather consider a splitting based on Jensen–Bregman divergences. This choice leads to a statistical model that can be reformulated as a double-level AXDA model, yielding the derivation of the proposed  Jensen–Bregman split Gibbs sampler (JB-SGS). Finally, by considering two particular statistical models for the likelihood, we demonstrate that JB-SGS defines a general sampling framework that encompasses several algorithms previously proposed in the literature. We emphasize that these existing methods were originally derived in an \emph{ad hoc} manner to overcome limitations of Q-SGS. On the contrary, the proposed JB-SGS arises from a principled and unified construction.

\section{Split Gibbs sampling with Jensen-Bregman coupling kernel}\label{sec:formulation}
\subsection{Single-level AXDA with Jensen-Bregman splitting}\label{sec:AXDA_JS}

Although initially designed from an empirical perspective, SGS owns theoretical foundations which can be anchored in the framework of AXDA \cite{vono2020asymptotically}. In a Bayesian setting, AXDA defines the surrogate distribution (aka the \emph{split} distribution) as
\begin{equation}\label{eq:augmented_distribution}
  \pi_{\rho}(\bx,\bz) \propto \exp\left[ - f(\bx;\by) - g(\bz)\right] \kappa_{\rho}(\bz,\bx)
\end{equation}
where the coupling kernel $\kappa_{\rho}(\cdot,\cdot)$ satisfies the following assumption.
\begin{assumption}\label{assump:kernel}
    The kernel $\kappa_\rho(\cdot, \bx)$ is defined such that $\int \kappa_{\rho}(\bz,\bx)d\bz < +\infty $  and weakly converges towards the Dirac measure at $\bx$ as $\rho \rightarrow 0$.
\end{assumption}
A key advantage of AXDA lies in the following property, satisfied by the marginal split distribution $\tilde{\pi}_{\rho}(\bx) \triangleq \int \pi(\bx,\bz) d\bz$.
\begin{theorem}[see \cite{vono2020asymptotically}]\label{prop:AXDA_TV}
    If Assumption \ref{assump:kernel} is satisfied, then
\begin{equation}\label{eq:convTV}
    \left\| \pi(\bx) - {\tilde{\pi}}_\rho(\bx) \right\|_{\mathrm{TV}} \xrightarrow[\rho \to 0]{} 0.
\end{equation}
\end{theorem}
The Q-SGS defined in \eqref{eq:SGS_canonical} can be interpreted as a conventional Gibbs sampler targeting the split distribution \eqref{eq:augmented_distribution} in the specific case of a Gaussian coupling kernel, i.e., $\kappa_{\rho}(\bz,\bx) \propto \exp\left(-\frac{1}{2\rho}\left\|\bx - \bz\right\|_2^2\right)$. Moreover, SGS produces samples $\left\{\bx^{(t)}\right\}_{t\geq 1}$ asymptotically distributed according to $\tilde{\pi}_{\rho}(\bx)$, which itself converges to the true posterior distribution $\pi(\bx)$ as the coupling parameter $\rho$ vanishes.

However, as stated above, choosing a Gaussian coupling kernel is unsuitable when handling non-Gaussian likelihoods. From a practical point-of-view, it leads to a likelihood step \eqref{eq:likelihood_step} that lacks conjugacy  turning it into a non-trivial sampling task. More fundamentally, such a quadratic coupling does not account for the geometry of the problem and possible corresponding constraints. A natural alternative is to define the coupling kernel through a Bregman divergence \cite{bregman1967relaxation}, whose properties make it appealing in this context. It defines admissible coupling kernels for AXDA \cite{vono2020asymptotically} and is closely connected to exponential family distributions \cite{banerjee2005clustering}. Yet, they are generally asymmetric, which questions the position of the two variables $\bx$ and $\bz$ in the kernel, which directly  impacts both the definition of the split distribution and the resulting Gibbs sampling scheme. To overcome this limitation, we  define a coupling kernel satisfying Assumption \ref{assump:kernel}  as
\begin{equation}\label{eq:JS_kernel}
    {\kappa}_{\rho}(\bz,\bx) = \exp\left[ -\tfrac{2}{\rho} {D}_{h}(\bz,\bx) + \psi_{\rho}(\bz)\right]
\end{equation}
where $\psi_{\rho}(\bz)$ is a normalizing potential and
\begin{equation}
    {D}_{h}(\bz,\bx) = \frac{1}{2}  \left[d_h\left(\bz,\frac{\bx+\bz}{2}\right) +d_{h}\left(\bx,\frac{\bx+\bz}{2}\right) \right]
\end{equation}
is the Jensen-Bregman (or Burbea-Rao) divergence, a symmetric extension of the Bregman divergence $d_h(\cdot,\cdot)$ generated by a strictly convex and differentiable function $h: \mathbb{R}^n \to \mathbb{R}$  \cite{nielsen2011burbea}.

\subsection{Equivalent double-level AXDA with Bregman splitting}
Devising a two-step SGS scheme similar to \eqref{eq:SGS_canonical} for the coupling kernel \eqref{eq:JS_kernel} is challenging mainly because of the complex conditional structure of the augmented distribution \eqref{eq:augmented_distribution}. Interestingly, a much more convenient conditional structure emerges by resorting to the following result.
\begin{lemma}\label{lemma:JS_kernel2}
    As $\rho$ vanishes, the Jensen-Bregman coupling kernel satisfies
\begin{equation}\label{eq:JS_kernel2}
    {\kappa}_{\rho}(\bz,\bx) \propto_{\bz} \int \kappa_{\rho}^{+}(\bz,\bu) \kappa_{\rho}^{-}(\bu,\bx) d\bu
\end{equation}
where the kernels $\kappa_{\rho}^{+}(\cdot,\cdot)$ and $\kappa_{\rho}^{-}(\cdot,\cdot)$ defined as
\begin{align}
    {\kappa}_{\rho}^{+}(\bz,\cdot) &\triangleq \exp\left[ -\tfrac{1}{\rho} {d}_{h}(\bz,\cdot) + \varphi_{\rho}^{+}(\bz) \right] \label{eq:B_kernel}\\
    {\kappa}_{\rho}^{-}(\bu,\cdot) &\triangleq \exp\left[ -\tfrac{1}{\rho} {d}_{h}(\cdot,\bu) + \varphi_{\rho}^{-}(\bu) \right] \label{eq:B_kernel_bis}
\end{align}
embed the (asymmetric) Bregman divergence $d_h(\cdot,\cdot)$ in both sides and $\varphi_{\rho}^{\pm}(\cdot)$ are normalizing potentials. 
\end{lemma}
\begin{proof}
    See Appendix.
\end{proof}
This lemma states that the Jensen-Bregman kernel \eqref{eq:JS_kernel} asymptotically arises from the marginalization of the product of two Bregman kernels under the same geometry over a third auxiliary variable $\bu \in \mathbb{R}^m$. By plugging \eqref{eq:B_kernel} into \eqref{eq:JS_kernel2} and \eqref{eq:augmented_distribution}, this result easily translates to the split distribution \eqref{eq:augmented_distribution}.
\begin{theorem}\label{th:augmented_distribution}
As $\rho$ vanishes, the split distribution $\pi_{\rho}(\bx,\bz)$ resulting from a Jensen-Bregman coupling kernel writes
\begin{equation}\label{eq:augmented_distribution3}
  \pi_{\rho}(\bx,\bz) \propto_{\bz} \int \pi_{\rho}(\bx,\bz,\bu) d\bu
\end{equation}
where $\pi_{\rho}(\bx,\bz,\bu)$ defines the augmented split distribution
\begin{multline}\label{eq:augmented_distribution2}
  \pi_{\rho}(\bx,\bz,\bu) \propto \exp\left[ - f(\bx;\by) - g(\bz) -\tfrac{1}{\rho}d_h\left(\bz,\bu\right) \right.\\
  \left.  + \varphi_{\rho}^{+}(\bz) - \tfrac{1}{\rho}d_{h}\left(\bx,\bu\right) + \varphi_{\rho}^{-}(\bu) \right] .
\end{multline}
\end{theorem}
The augmented split distribution $\pi_{\rho}(\bx,\bz,\bu)$ in \eqref{eq:augmented_distribution2} can be interpreted as the result of two successive AXDA of the target posterior \eqref{eq:posterior}, first with the kernel $\kappa_{\rho}^{+}(\bz,\bx)$ and then with the kernel $ \kappa_{\rho}^{-}(\bu,\bx)$. 
This double-level augmentation yields a conditional structure that naturally leads to the Jensen–Bregman split Gibbs sampler (JB-SGS) described below.

\begin{remark}
    Reverting to a single-level AXDA scheme would lead to a split distribution defined by a unique auxiliary variable $\bz$ and a Bregman coupling kernel $\kappa_{\rho}^{+}(\bz,\bx)$ or $\kappa_{\rho}^{-}(\bz,\bx)$. In the Euclidean geometry, by symmetry, the Q-SGS in \eqref{eq:SGS_canonical} shows that the conditional structure of the resulting split distributions could be manageable. However, depending on choice of the kernel, its naive extension to non-Gaussian likelihoods makes either the likelihood step \eqref{eq:likelihood_step} or the prior step \eqref{eq:prior_step} difficult to handle (see Section \ref{subsec:Poisson}). 
\end{remark}

\subsection{Sampling from the augmented split distribution}
The proposed JB-SGS algorithm consists in deriving a conventional Gibbs sampler targeting the augmented split distribution $\pi_{\rho}(\bx,\bz,\bu)$, according to the following 3-step procedure\footnote{The rewriting of the conditional distribution of the second auxiliary variable $\bu$ in \eqref{eq:BJSGS_u} leverages the parallelogram law (see also proof of Lemma \ref{lemma:JS_kernel2} in Appendix).}
\begin{subequations}\label{eq:BJSGS_canonical}
\begin{align}
    \bx &\sim \pi_{\rho}(\bx|\bu) \propto \exp\left[ - f(\bx;\by) - \tfrac{1}{\rho}d_h\left(\bx,\bu\right)\right] \label{eq:BJSGS_x} \\ %
    \bz &\sim \pi_{\rho}(\bz|\bu) \propto \exp\left[ - g(\bz) - \tfrac{1}{\rho}d_h\left(\bz,\bu\right) + \varphi_{\rho}^{+}(\bz) \right]\label{eq:BJSGS_z} \\ %
    \bu &\sim \pi_{\rho}(\bu|\bx,\bz) \propto \exp\left[ - \tfrac{2}{\rho}d_h \left(\tfrac{\bx+\bz}{2},\bu\right) + \varphi_{\rho}^{-}(\bu)\right]. \label{eq:BJSGS_u}
\end{align}
\end{subequations}
The sampling scheme of JB-SGS naturally deserves a careful examination of each of its three steps. 
\subsubsection{Likelihood step \eqref{eq:BJSGS_x}}\label{subsubsec:BJSGS_x}
To ease sampling, it clearly appears that the generator $h$ underlying the Bregman divergence $d_h(\bx,\bu)$ should be chosen such that this latter leads to a conjugate distribution, whenever it is possible. While selecting an appropriate generator may require some expertise in Bayesian modeling, useful guidelines can be easily drawn from the tight connections between Bregman divergences and  exponential family distributions \cite{nielsen2009statistical}. With a suitable choice of $h$, this likelihood step naturally preserves the geometry and constraints of the problem even for non-Gaussian likelihoods, unlike its Q-SGS counterpart \eqref{eq:likelihood_step}.
\subsubsection{Prior step \eqref{eq:BJSGS_z}}
The practical implementation of the  prior step highly depends on the regularization potential $g(\cdot)$. For complex priors, advanced sampling tools may be required. Yet, any practitioner should keep in mind that this prior step boils down to a stochastic variant of a Bregman proximal mapping \cite{bauschke2017convex}. As for Q-SGS, the proposed JB-SGS may embed a (proximal) Langevin algorithm, such as the  mirror sampling algorithm designed to account for non-Euclidean geometry \cite{klatzer2025efficient}. Interpreting \eqref{eq:BJSGS_z} as a stochastic denoising, an alternative approach is to adopt a PnP strategy to sample a denoised variable $\bz$ when the observation $\bu$ is affected by noise whose statistics is driven by the divergence $d_h(\cdot,\cdot)$.
\subsubsection{Consensus step \eqref{eq:BJSGS_u}} 
This final step leads JB-SGS  to depart significantly from the sampling scheme \eqref{eq:SGS_canonical} of Q-SGS. The conditional distribution $\pi_{\rho}(\bu|\bx,\bz)$ is solely defined by the divergence $d_h(\cdot,\cdot)$. It yields an easy-to-sample distribution from the exponential family for a wide range of generators $h$. Besides, it is worth noting that this distribution behaves such that $\mathrm{E}[\bu|\bx,\bz] \rightarrow \tfrac{\bx+\bz}{2}$ as $\rho \rightarrow 0$.  This shows that this step can be interpreted as a geometry-aware consensus where the auxiliary variable $\bu$ acts as a stochastic buffer mediating the interaction between $\bx$ and $\bz$.  Yet, such augmentations have been already promoted  to reduce the interactions between samples generated by MCMC algorithms \cite{besag1993spatial,higdon1998auxiliary,van2001art}, including Q-SGS (see Section \ref{subsec:application_gaussian}). 

\section{Revisiting classical statistical models}\label{sec:application}
This section illustrates the versatility of the proposed JB-SGS by elucidating its derivations for two statistical models of the likelihood $p(\by|\bx) \propto \exp[-f(\bx;\by)]$ often  encountered when addressing Bayesian inference problems. Interestingly, these two particular instances of JB-SGS correspond to sampling schemes that have previously appeared in the literature, though derived in an empirical or ad hoc manner. For sake of simplicity\footnote{When the data-fitting term also depends on a forward operator $\mathcal{A}(\cdot)$ such that $\mathrm{E}[\by|\bx] = \mathcal{A}(\bx)$, the proposed JB-SGS may require to be granted with dedicated strategies to ease the likelihood step \eqref{eq:BJSGS_x}.}, all considered models are restricted to the case where $m=n$, $\mathrm{E}[\by|\bx] = \bx$ and $\mathrm{cov}[\by|\bx] \propto \mathbf{I}_m$. In the context of Bayesian inverse problems, this boils down to performing a simple denoising task with independent and identically distributed observations $\left\{y_i\right\}_{i=1}^m$.

\subsection{Gaussian likelihood}\label{subsec:application_gaussian}
For Gaussian likelihoods, the data-fitting term takes the form $f(\bx;\by)=\tfrac{1}{2\sigma^2} \left\|\by-\bx\right\|_2^2$. Thus, it is sound to choose a quadratic generator $h(\bx)=\tfrac{1}{2}\left\|\bx\right\| ^2_2$, which induces the squared Euclidean distance $d_h(\bx,\bu) = \tfrac{1}{2}\left\|\bx-\bu\right\|_2^2$. Thanks to conjugacy, the likelihood step \eqref{eq:BJSGS_x} reduces to sampling from a multivariate  Gaussian distribution \cite{vono2022high}, i.e.,
\begin{equation*}
\bx|\by,\bu \sim \mathcal{N}\left(\boldsymbol{\mu},\boldsymbol{\Sigma}\right) \ \text{with} \ 
\left\{
    \begin{array}{ll}
   \boldsymbol{\mu} &= \boldsymbol{\Sigma}^{-1} \left(\tfrac{1}{\sigma^2} \by + \tfrac{1}{\rho} \bu\right) \\
    \boldsymbol{\Sigma}^{-1} &= \left(\tfrac{1}{\sigma^2} + \tfrac{1}{\rho}\right) \mathbf{I}_m.    
    \end{array}
\right.
\end{equation*}
The prior step \eqref{eq:BJSGS_z} consists in sampling from the distribution $\pi_{\rho}(\bz|\bu) \propto \exp\left( - g(\bz) - \tfrac{1}{2\rho}\left\|\bz-\bu\right\|_2^2 \right)$. When this latter is not a known and easy-to-sample distribution, it requires a restricted Gaussian oracle (RGO) \cite{lee2021structured} whose design depends on the characteristics of the regularization potential $g(\cdot)$. In the simplest case of a smooth and log-concave potential, this RGO can be limited to performing one step of a Langevin diffusion, as in \cite{Faye2024,Faye_2025sub}. In a more general yet challenging case of a non-smooth and non-convex but proximable regularization potential $g(\cdot)$, proximal samplers are required \cite{pereyra2016proximal,durmus2018efficient,liang2023a}. Conversely, adopting a PnP strategy, this prior step can be performed using a stochastic denoiser such as a pretrained diffusion-based model, as proposed in \cite{Coeurdoux2024pnp} and successfully used in recent works \cite{meier2025clustering,gao2025diffusion,zhang2026bayesian}. In the case of an inexact prior step, the generated samples are drawn from a stationary distribution with bias, which  can be corrected using a Metropolis-Hastings procedure \cite{robert1999monte}. 
Finally, the consensus step \eqref{eq:BJSGS_u} straightforwardly samples from a Gaussian distribution, i.e., $\bu |\bx,\bz \sim \mathcal{N}(\tfrac{\bx+\bz}{2},\tfrac{2}{\rho} \mathbf{I}_m)$. 

Here, it is worth emphasizing that the Bregman coupling kernel is already symmetric defined through a squared Euclidean distance, i.e., $\kappa_{\rho}^{\pm}(\bx,\bz) \propto \exp\left[-\tfrac{1}{2\rho} \left\|\bx-\bz\right\|_2^2\right]$. Thus, in this simple Gaussian setting, the associated Jensen-Bregman kernel reduces to the Bregman kernel, i.e., $\kappa_{\rho}(\bx,\bz) = \kappa_{\rho}^{\pm}(\bx,\bz)$, which means that JB-SGS and Q-SGS target the same split distribution. This equivalence might raise the question of the actual benefit of the 3-step JB-SGS sampling scheme \eqref{eq:BJSGS_canonical} over the seemingly simpler 2-step Q-SGS  \eqref{eq:SGS_canonical} when working in an Euclidean geometry. One answer is brought by noting that the sampling scheme \eqref{eq:BJSGS_canonical} matches the \emph{split-and-augmented} Gibbs sampler (Q-SAGS) proposed in \cite{Vono2019split } as an extension of Q-SGS. In that work, the introduction of the second auxiliary variable $\bu$ was empirically motivated to improve the mixing properties of Q-SGS. Indeed,  when comparing the performance of Q-SGS and Q-SAGS on various large-scale imaging inverse problems, such improvements were notably observed by monitoring the autocorrelation functions of the resulting Markov chains. It is worth noting that  the main difference between Q-SAGS and Q-SGS lies in the additional simple consensus step \eqref{eq:BJSGS_u}. Thus, the proposed JB-SGS (or equivalently Q-SAGS) can be systematically derived in place of Q-SGS to improve sampling, at the price of a negligible computational overhead.

\subsection{Poisson likelihood}\label{subsec:Poisson}
Poisson likelihoods naturally arise in counting processes, such as photon-limited imaging \cite{melidonis2023efficient}. In this case, the data-fitting term is defined as a generalized Kullback-Leibler divergence $f(\bx;\by) = \sum_{i=1}^m x_i - y_i \log x_i$. Choosing a divergence induced by a quadratic generator as in Section \ref{subsec:application_gaussian} would lead to a geometry-incompliant and computationally demanding likelihood step \cite{vono2019bayesian}. Instead, as recommended in Section \ref{subsubsec:BJSGS_x}, the generator $h$ should be chosen to preserve the underlying geometry and to ease the subsequent computations. These requirements are naturally satisfied by the Burg entropy $h(\bx)=-\sum_{i=1}^m \log x_i$ defined for $\bx \in \mathbb{R}^m_{++}$, resulting in coupling kernels $\kappa^{\pm}_\rho(\cdot,\cdot)$ defined through the Itakura-Saito divergence $d_h(\bx,\bu) = \sum_{i=1}^m \left(  \tfrac{x_i}{u_i} - \log \tfrac{x_i}{u_i} -1 \right)$. Straightforward computations show that JB-SGS is characterized by simple likelihood and consensus steps given by
\begin{align}
    \bx |\by,\bu &\sim \prod_{i=1}^m \mathcal{G}\left(y_i + \frac{1}{\rho}+1, 1+\frac{1}{\rho u_i}\right)\\
    \bu |\bz,\bx &\sim \prod_{i=1}^m \mathcal{IG}\left(\frac{2}{\rho}, \frac{x_i+z_i}{\rho}\right)
\end{align}
where $\mathcal{G}(\cdot,\cdot)$ and  $\mathcal{IG}(\cdot,\cdot)$ stand for Gamma and inverse-Gamma distributions, respectively. The prior step consists in sampling from the conditional distribution $\pi_{\rho}(\bz|\bu) \propto \exp\left[- E(\bz,\bu)\right]$ whose overall potential writes
\begin{equation}\label{eq:potential_JBSGS}
 E(\bz,\bu) = g(\bz) + \sum_{i=1}^m \log z_i + \frac{1}{\rho} \sum_{i=1}^m \left(\frac{z_i}{u_i} - \log z_i\right).
\end{equation}
This JB-SGS sampling scheme coincides with the one proposed in \cite{faye2025bregman} for  Poisson inverse imaging problems. Yet, in that work, the authors followed a reciprocal derivation, starting with a single-level AXDA with a Bregman kernel $\kappa_\rho^+(\bz,\bx)$, subsequently complemented by a second level  with a Bregman kernel $\kappa_{\rho}^-(\bu,\bx)$. One motivation of this second splitting was to reach a more amenable prior step. Indeed, the authors of \cite{faye2025bregman} showed that, when reverting to a single-level AXDA induced by a Bregman kernel\footnote{The kernel $\kappa_\rho^-(\bz,\bx)$ has been chosen  to preserve conjugacy in the likelihood step. Adopting the kernel $\kappa_\rho^+(\bz,\bx)$ would lead to a potential $E(\bz,\bu)$ similar to \eqref{eq:potential_JBSGS} but with a more computationally demanding  likelihood step.} $\kappa_{\rho}^-(\bz,\bx)$, the SGS prior step involves a potential defined by
\begin{equation}\label{eq:potential_SGS}
 E(\bz,\bu) = g(\bz) + \left(1+\frac{1}{\rho}\right) \sum_{i=1}^m \log z_i + \frac{1}{\rho} \sum_{i=1}^m \left(\frac{x_i}{z_i}\right).
\end{equation}
It clearly appears that this potential owns much less favorable regularity properties with respect to $\bz$ than the potential \eqref{eq:potential_JBSGS} exhibited by JB-SGS.  This finding  illustrates once again the value of considering a Bregman-Jensen kernel. One alternative would be to reverting to the equivalent single-level AXDA formulation of JB-SGS. The split distribution would be defined through a Bregman-Jensen coupling kernel defined by the symmetrized  Itakura-Saito divergence
\begin{equation}
 D_h(\bz,\bx) = \sum_{i=1}^m \log \left(\frac{z_i+x_i}{2\sqrt{x_i z_i}{}}\right).
\end{equation}
Then the corresponding SGS exhibits a likelihood step lacking conjugacy properties, defined by the product of the Poisson likelihood and a scaled Beta prime distribution \cite{Perez2017}. Finally, these insights show that the proposed JB-SGS combines all the practical values expected from splitting-based samplers.

\section{Conclusion}
This letter introduced a generalized split Gibbs sampling (SGS) framework for non-Euclidean geometries using a Jensen-Bregman divergence-based coupling kernel. This approach extends the applicability of SGS beyond Gaussian likelihoods, providing a unified and principled foundation for existing empirical methods. This work opens new avenues for Bayesian estimation in complex geometric spaces, with potential applications in inverse problems and beyond.

\section*{Appendix: Sketch of proof of Lemma \ref{lemma:JS_kernel2}}\label{app:JS_kernel2}
Resorting to the parallelogram law,  the Jensen-Bregman divergence $D_{h}(\bz,\bx)$ can be rewritten \cite{nielsen2021geodesic}
\begin{equation*}
    D_{h}(\bz,\bx) = \frac{1}{2}  \left[d_h\left(\bz,\bu\right) +d_{h}\left(\bx,\bu\right) - 2d_h\left(\frac{\bx+\bz}{2},\bu\right)\right].
\end{equation*}
for any auxiliary variable $\bu\in \mathbb{R}^m$. After  incorporating the normalizing potentials $\varphi^{-}(\bu)$ and $\varphi^{+}(\bz)$, we can write 
\begin{eqnarray*}
\exp\left[-\tfrac{2}{\rho}D_{h}(\bz,\bx)+\varphi_\rho^{+}(\bz)\right]
\exp\left[-\tfrac{2}{\rho}d_h\left(\tfrac{\bx+\bz}{2},\bu\right)  + \varphi_\rho^{-}(\bu) \right] \\
= \exp\left[-\tfrac{1}{\rho} d_h\left(\bz,\bu\right) - \tfrac{1}{\rho} d_{h}\left(\bx,\bu\right) + \varphi_\rho^{+}(\bz) + \varphi_\rho^{-}(\bu) \right].
\end{eqnarray*}
By setting $\psi_\rho(\cdot)=\varphi_\rho^{+}(\cdot)$ in \eqref{eq:JS_kernel}, it yields
\begin{equation}\label{eq:kernel_product}
    \kappa_\rho(\bz,\bx)\, \kappa_{\rho/2}^{-}\left(\bu,\tfrac{\bz+\bx}{2}\right) =
    \kappa_\rho^+\left(\bz,\bu\right)\, \kappa_\rho^-\left(\bu,\bx\right)
\end{equation}
Using a Laplace's approximation of $\exp \varphi_{\rho/2}^-(\bu)$ defining $\kappa_{\rho/2}^{-}\left(\bu,\cdot\right)$, when $\rho$ vanishes, the coupling kernel ${\kappa}_{\rho/2}^{-}(\bu,\cdot) $ behaves as $\rho^{-1/2}\exp\left[ -\tfrac{2}{\rho} {d}_{h}(\cdot,\bu) + \phi(\bu) \right]$ which is a kernel satisfying Assumption \ref{assump:kernel}. Then,  integrating with respect to $\bu$ on both sides of \eqref{eq:kernel_product} yields the result \eqref{eq:JS_kernel2} when $\rho\rightarrow 0$.

\bibliographystyle{IEEEtran}

\bibliography{strings_all_ref,biblio}

\begin{thebibliography}{10}
\providecommand{\url}[1]{#1}
\csname url@samestyle\endcsname
\providecommand{\newblock}{\relax}
\providecommand{\bibinfo}[2]{#2}
\providecommand{\BIBentrySTDinterwordspacing}{\spaceskip=0pt\relax}
\providecommand{\BIBentryALTinterwordstretchfactor}{4}
\providecommand{\BIBentryALTinterwordspacing}{\spaceskip=\fontdimen2\font plus
\BIBentryALTinterwordstretchfactor\fontdimen3\font minus
  \fontdimen4\font\relax}
\providecommand{\BIBforeignlanguage}[2]{{%
\expandafter\ifx\csname l@#1\endcsname\relax
\typeout{** WARNING: IEEEtran.bst: No hyphenation pattern has been}%
\typeout{** loaded for the language `#1'. Using the pattern for}%
\typeout{** the default language instead.}%
\else
\language=\csname l@#1\endcsname
\fi
#2}}
\providecommand{\BIBdecl}{\relax}
\BIBdecl

\bibitem{Vono2019split}
M.~Vono, N.~Dobigeon, and P.~Chainais, ``Split-and-augmented {G}ibbs sampler --
  {A}pplication to large-scale inference problems,'' \emph{IEEE Trans. Signal
  Process.}, vol.~67, no.~6, pp. 1648--1661, 2019.

\bibitem{vono2022efficient}
M.~Vono, D.~Paulin, and A.~Doucet, ``Efficient {MCMC} sampling with
  dimension-free convergence rate using {ADMM}-type splitting,'' \emph{J. Mach.
  Learning Research}, vol.~23, no.~1, pp. 1100--1168, 2022.

\bibitem{neal2011distributed}
P.~Neal, C.~Eric, P.~Borja, and E.~Jonathan, ``Distributed optimization and
  statistical learning via the alternating direction method of multipliers,''
  \emph{Foundations and Trends{\textregistered} in Machine learning}, vol.~3,
  no.~1, pp. 1--122, 2011.

\bibitem{rendell2020global}
L.~J. Rendell, A.~M. Johansen, A.~Lee, and N.~Whiteley, ``Global consensus
  {M}onte {C}arlo,'' \emph{J. Comput. Graph. Stat.}, vol.~30, no.~2, pp.
  249--259, 2020.

\bibitem{pereyra2023split}
M.~Pereyra, L.~A. Vargas-Mieles, and K.~C. Zygalakis, ``The split {G}ibbs
  sampler revisited: improvements to its algorithmic structure and augmented
  target distribution,'' \emph{SIAM J. Imag. Sci.}, vol.~16, no.~4, pp.
  2040--2071, 2023.

\bibitem{plassier2021dg}
V.~Plassier, M.~Vono, A.~Durmus, and E.~Moulines, ``{DG-LMC}: a turn-key and
  scalable synchronous distributed {MCMC} algorithm via {L}angevin {M}onte
  {C}arlo within {G}ibbs,'' in \emph{Proc. Int. Conf. Machine Learning
  (ICML)}.\hskip 1em plus 0.5em minus 0.4em\relax PMLR, 2021, pp. 8577--8587.

\bibitem{Faye_2025sub}
\BIBentryALTinterwordspacing
E.~C. Faye, M.~D. Fall, S.~Delchini, and N.~Dobigeon, ``Bridging data-driven
  priors via the score function for posterior sampling – {C}omparative review
  and experimental study,'' \emph{arXiv preprint arXiv:2606.14800}, 2026.
  [Online]. Available: \url{https://arxiv.org/abs/2606.14800/}
\BIBentrySTDinterwordspacing

\bibitem{vono2018sparse}
M.~Vono, N.~Dobigeon, and P.~Chainais, ``Sparse {B}ayesian binary logistic
  regression using the split-and-augmented {G}ibbs sampler,'' in \emph{Proc.
  IEEE Workshop Mach. Learning for Signal Process. (MLSP)}, 2018, pp. 1--6.

\bibitem{Coeurdoux2024pnp}
F.~Coeurdoux, N.~Dobigeon, and P.~Chainais, ``Plug-and-play split {G}ibbs
  sampler: embedding deep generative priors in {B}ayesian inference,''
  \emph{IEEE Trans. Image Process.}, vol.~33, pp. 3496--3507, May 2024.

\bibitem{bouman2023generative}
C.~A. Bouman and G.~T. Buzzard, ``Generative plug and play: Posterior sampling
  for inverse problems,'' in \emph{Ann. Allerton Conf. Comm. Control Computing
  (Allerton)}, 2023, pp. 1--7.

\bibitem{wu2024principled}
Z.~Wu, Y.~Sun, Y.~Chen, B.~Zhang, Y.~Yue, and K.~Bouman, ``Principled
  probabilistic imaging using diffusion models as plug-and-play priors,'' in
  \emph{Adv. in Neural Information Process. Systems (NeurIPS)}, vol.~37, 2024,
  pp. 118\,389--118\,427.

\bibitem{Sun2024provable}
Y.~Sun, Z.~Wu, Y.~Chen, B.~T. Feng, and K.~L. Bouman, ``Provable probabilistic
  imaging using score-based generative priors,'' \emph{IEEE Trans. Comput.
  Imag.}, 2024.

\bibitem{Faye2024}
E.~C. Faye, M.~D. Fall, and N.~Dobigeon, ``Regularization by denoising:
  {B}ayesian model and {L}angevin-within-split {G}ibbs sampling,'' \emph{IEEE
  Trans. Image Process.}, vol.~34, pp. 221--234, Jan. 2025.

\bibitem{vono2022high}
M.~Vono, N.~Dobigeon, and P.~Chainais, ``High-dimensional {G}aussian sampling:
  a review and a unifying approach based on a stochastic proximal point
  algorithm,'' \emph{SIAM Rev.}, vol.~64, no.~1, pp. 3--56, 2022.

\bibitem{faye2025bregman}
\BIBentryALTinterwordspacing
E.~C. Faye, M.~D. Fall, N.~Dobigeon, and E.~Barat, ``Bregman geometry-aware
  split {G}ibbs sampling for {B}ayesian {P}oisson inverse problems,''
  \emph{arXiv preprint arXiv:2511.12257}, 2025. [Online]. Available:
  \url{https://arxiv.org/abs/2511.12257/}
\BIBentrySTDinterwordspacing

\bibitem{vono2020asymptotically}
M.~Vono, N.~Dobigeon, and P.~Chainais, ``Asymptotically exact data
  augmentation: Models, properties, and algorithms,'' \emph{J. Comput. Graph.
  Stat.}, vol.~30, no.~2, pp. 335--348, 2020.

\bibitem{bregman1967relaxation}
L.~M. {B}regman, ``The relaxation method of finding the common point of convex
  sets and its application to the solution of problems in convex programming,''
  \emph{USSR computational mathematics and mathematical physics}, vol.~7,
  no.~3, pp. 200--217, 1967.

\bibitem{banerjee2005clustering}
A.~Banerjee, S.~Merugu, I.~S. Dhillon, and J.~Ghosh, ``Clustering with
  {B}regman divergences,'' \emph{J. Mach. Learning Research}, vol.~6, no. Oct,
  pp. 1705--1749, 2005.

\bibitem{nielsen2011burbea}
F.~Nielsen and S.~Boltz, ``The {B}urbea-{R}ao and {B}hattacharyya centroids,''
  \emph{IEEE Trans. Inf. Theory}, vol.~57, no.~8, pp. 5455--5466, 2011.

\bibitem{nielsen2009statistical}
F.~Nielsen and V.~Garcia, ``Statistical exponential families: A digest with
  flash cards,'' \emph{arXiv preprint arXiv:0911.4863}, 2009.

\bibitem{bauschke2017convex}
H.~H. Bauschke and P.~L. Combettes, \emph{Convex Analysis and Monotone Operator
  Theory in Hilbert Spaces}, 2nd~ed., ser. CMS Books in Mathematics.\hskip 1em
  plus 0.5em minus 0.4em\relax Springer, 2017.

\bibitem{klatzer2025efficient}
T.~Klatzer, S.~Melidonis, M.~Pereyra, and K.~C. Zygalakis, ``Efficient
  {B}ayesian computation using plug-and-play priors for {P}oisson inverse
  problems,'' \emph{SIAM J. Imag. Sci.}, vol.~19, no.~2, pp. 1325--1363, June
  2026.

\bibitem{besag1993spatial}
J.~Besag and P.~J. Green, ``Spatial statistics and {B}ayesian computation,''
  \emph{J. R. Stat. Soc. Series B: Stat. Methodol.}, vol.~55, no.~1, pp.
  25--37, 1993.

\bibitem{higdon1998auxiliary}
D.~M. Higdon, ``Auxiliary variable methods for {M}arkov chain {M}onte {C}arlo
  with applications,'' \emph{J. Amer. Stat. Assoc.}, vol.~93, no. 442, pp.
  585--595, 1998.

\bibitem{van2001art}
D.~A. Van~Dyk and X.-L. Meng, ``The art of data augmentation,'' \emph{J.
  Comput. Graph. Stat.}, vol.~10, no.~1, pp. 1--50, 2001.

\bibitem{lee2021structured}
Y.~T. Lee, R.~Shen, and K.~Tian, ``Structured logconcave sampling with a
  restricted {G}aussian oracle,'' in \emph{Proc. Conf. Learning Theory}.\hskip
  1em plus 0.5em minus 0.4em\relax PMLR, 2021, pp. 2993--3050.

\bibitem{pereyra2016proximal}
M.~Pereyra, ``Proximal {M}arkov chain {M}onte {C}arlo algorithms,'' \emph{Stat.
  Comput.}, vol.~26, pp. 745--760, 2016.

\bibitem{durmus2018efficient}
A.~Durmus, E.~Moulines, and M.~Pereyra, ``Efficient {B}ayesian computation by
  proximal {M}arkov chain {M}onte {C}arlo: when {L}angevin meets {M}oreau,''
  \emph{SIAM J. Imag. Sci.}, vol.~11, no.~1, pp. 473--506, 2018.

\bibitem{liang2023a}
J.~Liang and Y.~Chen, ``A proximal algorithm for sampling,'' \emph{Trans. Mach.
  Learning Research}, 2023.

\bibitem{meier2025clustering}
D.~Meier, S.~Yu, S.~Nandy, P.~Ghosal, and K.~Gan, ``Clustering by denoising:
  Latent plug-and-play diffusion for single-cell data,'' in \emph{Proc. IEEE
  Int. Conf. Learn. Represent. (ICLR)}, 2026.

\bibitem{gao2025diffusion}
H.~Gao, T.~Huang, L.~Guo, J.~He, and Y.~C. Eldar, ``Diffusion-prior split
  {G}ibbs sampling for synthetic aperture radar imaging under incomplete
  measurements,'' \emph{arXiv preprint arXiv:2512.02768}, 2025.

\bibitem{zhang2026bayesian}
Y.~Zhang, R.~Guo, and Y.~C. Eldar, ``Bayesian signal component decomposition
  via diffusion-within-{G}ibbs sampling,'' \emph{arXiv preprint
  arXiv:2602.10792}, 2026.

\bibitem{robert1999monte}
C.~P. Robert, G.~Casella, and G.~Casella, \emph{{M}onte {C}arlo statistical
  methods}.\hskip 1em plus 0.5em minus 0.4em\relax Springer, 1999, vol.~2.

\bibitem{melidonis2023efficient}
S.~Melidonis, P.~Dobson, Y.~Altmann, M.~Pereyra, and K.~Zygalakis, ``Efficient
  {B}ayesian computation for low-photon imaging problems,'' \emph{SIAM J. Imag.
  Sci.}, vol.~16, no.~3, pp. 1195--1234, 2023.

\bibitem{vono2019bayesian}
M.~Vono, N.~Dobigeon, and P.~Chainais, ``Bayesian image restoration under
  {P}oisson noise and log-concave prior,'' in \emph{Proc. IEEE Int. Conf.
  Acoust., Speech and Signal Process. (ICASSP)}.\hskip 1em plus 0.5em minus
  0.4em\relax IEEE, 2019, pp. 1712--1716.

\bibitem{Perez2017}
M.-E. P{\'e}rez, L.~R. Pericchi, and I.~C. Ram{\'i}rez, ``{The Scaled {B}eta2
  Distribution as a Robust Prior for Scales},'' \emph{Bayesian Analysis},
  vol.~12, no.~3, pp. 615 -- 637, 2017.

\bibitem{nielsen2021geodesic}
F.~Nielsen, ``On geodesic triangles with right angles in a dually flat space,''
  in \emph{Progress in information geometry: theory and applications}.\hskip
  1em plus 0.5em minus 0.4em\relax Springer, 2021, pp. 153--190.

\end{thebibliography}

\end{document}